\documentclass[conference]{IEEEtran}
\IEEEoverridecommandlockouts

\makeatletter
\def\endthebibliography{%
  \def\@noitemerr{\@latex@warning{Empty `thebibliography' environment}}%
  \endlist
}
\makeatother

\usepackage[cmex10]{amsmath}
\usepackage{amssymb,amsthm}

\usepackage{bm}
\usepackage{amsfonts}
\usepackage{relsize}
\usepackage{cleveref}
\usepackage{cuted}
\usepackage{subfig}
\usepackage{graphicx}
\usepackage{epstopdf}
\usepackage{algorithm}
\usepackage[noend]{algpseudocode}
\usepackage{bbm}
\usepackage{comment}
\usepackage{multirow}
\usepackage{pbox}
\usepackage{array}
\usepackage{booktabs}
\usepackage{multicol}
\usepackage{url}

\newtheorem{lemma}{Lemma}

\newcommand*\diff{\mathop{}\!\mathrm{d}}

\usepackage{textcomp}
\usepackage{cite}
\usepackage{subfig}
\usepackage{xcolor}
\def\BibTeX{{\rm B\kern-.05em{\sc i\kern-.025em b}\kern-.08em
    T\kern-.1667em\lower.7ex\hbox{E}\kern-.125emX}}
\usepackage{enumitem} 
\usepackage{capt-of}

\algnewcommand{\Inputs}[1]{%
  \State \textbf{Inputs:}
  \Statex \hspace*{\algorithmicindent}\parbox[t]{.8\linewidth}{\raggedright #1}
}
\algnewcommand{\Initialize}[1]{%
  \State \textbf{Initialize:}
  \Statex \hspace*{\algorithmicindent}\parbox[t]{.8\linewidth}{\raggedright #1}
}

\makeatletter
\def\BState{\State\hskip-\ALG@thistlm}
\makeatother
\title{Discriminative Mutual Information Estimation for the Design of Channel Capacity Driven Autoencoders}

\author{\IEEEauthorblockN{Nunzio A. Letizia and Andrea M. Tonello} \IEEEauthorblockA{\textit{Institute of Networked and Embedded Systems, Embedded Communication Systems Lab} \\
\textit{University of Klagenfurt, Austria}\\
Email: \{nunzio.letizia, andrea.tonello\}@aau.at}

}


\IEEEoverridecommandlockouts

\begin{document}

\maketitle


\begin{abstract}
The development of optimal and efficient machine learning-based communication systems is likely to be a key enabler of beyond 5G communication technologies. In this direction, physical layer design has been recently reformulated under a deep learning framework where the autoencoder paradigm foresees the full communication system as an end-to-end coding-decoding problem. Given the loss function, the autoencoder jointly learns the coding and decoding optimal blocks under a certain channel model. Because performance in communications typically refers to achievable rates and channel capacity, the mutual information between channel input and output can be included in the end-to-end training process, thus, its estimation becomes essential. 

In this paper, we present a set of novel discriminative mutual information estimators and we discuss how to exploit them to design capacity-approaching codes and ultimately estimate the channel capacity.
\end{abstract}


\maketitle
\section{Introduction}
\label{sec:introduction}
Recent advancements in machine learning for communications have fostered the design of systems with a data-driven paradigm. As a consequence, physical layer design has been reinterpreted using machine learning techniques to improve the coding and decoding scheme performance \cite{Oshea2017, Nachmani2018, Dorner2018, Alberge2019, Stark2019}. In particular, the work in \cite{Oshea2017} introduced the concept of autoencoder-based communication systems. In contrast to traditional bottom-up approaches, the encoder and decoder blocks can be jointly learned during the end-to-end autoencoder training phase \cite{Oshea2017}. The autoencoder is a deep neural network that maps a sequence of input bits $\mathbf{s}$ into a sequence of output bits $\hat{\mathbf{s}}$. The input sequence is generally transformed into data symbols $x$ that are then fed into an intermediate channel layer (or network), which introduces constraints, distortions and uncertainties. Given the intermediate channel model, the autoencoder is typically trained by minimizing the cross-entropy loss function, so that it essentially performs a classification task. However, it is known that training a classifier via cross-entropy often suffers from overfitting issues, especially for large networks \cite{Zhang2017}. Moreover, autoencoders for communications shall consider the channel capacity during the learning process in order to produce optimal channel input samples (latent codes). In this direction, the work in \cite{Letizia2021} proposed to use a mutual information regularizer to control and estimate the amount of information stored in the latent representation. In addition, it proposed to use label smoothing to improve the network decoding ability. In \cite{Letizia2021}, the mutual information was computed with the neural estimator MINE \cite{Mine2018} by inserting an appropriate block in the autoencoder architecture. MINE was also exploited in \cite{Wunder2019} to study optimal coding schemes when no channel model is attainable. Nevertheless, MINE suffers from high bias and variance. Thus, new mutual information estimators have been recently proposed to subvert such limitations \cite{Song2020, Poole2019a, LetiziaNIPS}. Among them, a promising approach is given by the discriminative mutual information estimator (DIME) \cite{LetiziaNIPS}, a neural network that directly estimates the density ratio
\begin{equation}
\label{eq:density_ratio}
R(\mathbf{x},\mathbf{y}) = \frac{p_{XY}(\mathbf{x},\mathbf{y})}{p_X(\mathbf{x})\cdot p_Y(\mathbf{y})}
\end{equation}
instead of the individual densities in \eqref{eq:density_ratio}. Concurrently, it is expected that novel and optimal channel coding techniques based on mutual information learning, estimation and maximization can significantly impact the development of beyond 5G communication technologies.
 
In this paper, inspired by the DIME estimator and the $f$-GAN training objectives \cite{Nowozin2016}, we firstly propose a new family of estimators referred to as $f$-DIME. We secondly propose $\gamma$-DIME, a family of estimators that can be used in the end-to-end autoencoder training process to target the channel capacity. We then include the developed estimators in the capacity-driven autoencoder proposed in \cite{Letizia2021} and evaluate their performance in terms of block-error-rate (BLER) and accuracy of the mutual information estimation w.r.t. to MINE.

The rest of the paper is organized as follows. Section \ref{sec:autoencoders} revisits the autoencoder-based communication systems and discusses the major advantages of using a mutual information regularization term in the loss function. Section \ref{sec:theory} reviews some variational lower bounds on the mutual information and the related estimators. Section \ref{sec:f-DIME} presents a new set of discriminative estimators. Section \ref{sec:results} compares the estimators and illustrates the results. Finally, conclusions are reported.

\section{Capacity-driven Autoencoders}
\label{sec:autoencoders}
Any communication system consists of three main blocks: the transmitter, the channel, and the receiver. The baseband transmitter maps a message $s\in \mathcal{M} = \{1,2,\dots,M\}$ into $n$ complex symbols $x \in \mathbb{C}^n$. The symbols are transmitted over the channel at a rate $R = \log_2(M)/n$ (bits per channel use), typically under a power constraint. The channel introduces impairments and modifies $x$ into a distorted and noisy version $y$. The task of the receiver is to produce an estimate $\hat{s}$ of the original $s$ from the channel output $y$. The idea proposed in \cite{Oshea2017} is to consider the full communication chain as an autoencoder, a single deep neural network optimized in an end-to-end fashion.
The encoder component $F(s;\theta_E)$ acts as a symbol modulator and maps $s$ into the channel input $x$; it is represented by a neural network of parameters $\theta_E$. The channel is implemented with a set of layers that attempts to approximate the conditional transition probability $p_Y(y|x)$. For complex channel models, the generator of GANs \cite{Goodfellow2014, Oshea2018, Righini2019, Letizia2019a} mimics the channel block with a pre-trained stochastic neural network $y = H(x;\theta_H)$ and allows back-propagated gradients. The decoder block $G(y;\theta_D)$ acts as a demodulator and maps the received samples into the estimate of $\hat{s}$ and the associated a posteriori probability $p_{\hat{S}|Y}$; it consists of a neural network of parameters $\theta_D$ with a softmax layer as output. During the training process, the autoencoder parameters $(\theta_E,\theta_D)$ are jointly updated in order to minimize the categorical cross-entropy loss function
\begin{equation}
\mathcal{L}(\theta_E,\theta_D) = \mathbb{E}_{(x,y)\sim p_{XY}(x,y)}\biggl[-\log\bigl(p_{\hat{S}|Y}(\hat{s}|y;\theta_D)\bigr)\biggr].
\end{equation}
The decoder performs a classification task as it tries to estimate the transmitted message among $M$ possibilities. However, the cross-entropy loss function is prone to overfitting issues \cite{Zhang2017} and it does not guarantee any optimality in the code and constellation design. In \cite{Letizia2021}, a twofold solution to mitigate such issues was proposed. Indeed, the label smoothing regularization technique was used to penalize confident estimations and improve the autoencoder generalization ability. To tackle the optimality issue, it was proposed to use a mutual information regularizer term computed between the channel input and output samples. Therein, it has been shown that both regularizers can be viewed as entropy penalty regularizers. The latter, in particular, was introduced to ultimately target the channel capacity. The overall proposed loss function reads as follows
\begin{align}
\label{eq:loss_function}
\hat{\mathcal{L}}(\theta_E, \theta_D) = \; & \mathbb{E}_{(s,y)\sim \hat{p}_{\hat{S}|Y}(\hat{s}|y)\cdot p_{Y}(y|x)}\biggl[-\log\bigl(p_{\hat{S}|Y}(\hat{s}|y;\theta_D)\bigr)\biggr] \nonumber \\
& -\beta I(X;Y),
\end{align}
where $\hat{p}_{\hat{S}|Y}(\hat{s}|y)=(1-\epsilon)\delta_{\hat{S},S}+\epsilon u(\hat{s})$ is the smoothed target distribution, with $u(\hat{s}) = 1/M$ and $\beta, \epsilon$ are positive regularization parameters.

From \eqref{eq:loss_function} it is interesting to notice that in each training iteration, the mutual information back-propagated gradient influences the encoder parameters (thus, the channel input distribution) and the energy of such gradient depends on $\beta$. The cross-entropy gradient, instead, guides both encoder and decoder towards invertible latent representations (in the limit of the channel distortion). Ideally in the non-parametric limit, a minimization of \eqref{eq:loss_function} corresponds to a maximization of the mutual information, thus, to channel capacity for a memory-less channel
\begin{equation}
C =\max_{p_X(x)} I(X;Y),
\end{equation}
where the channel input distribution $p_X(x)$ is subject to given power constraints. Hence, it is clear that an accurate and stable estimation of the mutual information $I(X;Y)$ is instrumental for the design of optimal channel coding and end-to-end communication schemes. Recently, deep neural networks have been leveraged to maximize variational lower bounds on the mutual information \cite{Nguyen2010,Mine2018, Poole2019a, LetiziaNIPS}. In particular, the use of discriminative variational lower bounds offer a promising methodology as they attempt to directly estimate the density ratio in \eqref{eq:density_ratio}.

\section{Discriminative Mutual Information Estimators}
\label{sec:theory}
The mutual information $I(X;Y)$ quantifies the statistical dependence between two random variables, $X$ and $Y$, by measuring the amount of information of $X$ via observation of $Y$. Unlike the correlation, the mutual information measures linear and non-linear associations \cite{Letizia2020} since its estimation requires the knowledge of three distributions
\begin{equation}
I(X;Y) = \mathbb{E}_{(x,y)\sim p_{XY}(x,y)}\biggl[\log\frac{p_{XY}(x,y)}{p_X(x)p_Y(y)}\biggr],
\end{equation}
where $p_X(x)$ and $p_Y(y)$ are the distributions of $X$ and $Y$, respectively, and $p_{XY}(x,y)$ is the joint probability distribution. The presence of these three components renders the mutual information estimation a challenging task and several approaches have been proposed to tackle the problem. Standard and traditional approaches rely on binning, density and kernel estimation \cite{Moon1995} and $k$-nearest neighbours \cite{Kraskov2004}. However, for applications such as channel coding, high-dimensional data is often involved and classic approaches fail to scale. Moreover, they can not easily be integrated in an optimization framework that targets the channel capacity. To overcome such limitations, neural networks can be used to obtain a parametric representations of the variational lower bounds on the mutual information, where the objective of the network is to estimate the density ratio. 

One of the first examples of neural estimators is MINE \cite{Mine2018}. The idea behind MINE is to exploit the Donsker-Varadhan dual representation of the Kullback-Leibler (KL) divergence to produce the bound to optimize
\begin{align}
\label{eq:MINE}
I(X;Y) \geq & I_{MINE}(X;Y) \nonumber \\
= & \sup_{\theta \in \Theta} \mathbb{E}_{(x,y)\sim p_{XY}(x,y)}[T_{\theta}(x,y)]  \nonumber \\
& - \log(\mathbb{E}_{(x,y)\sim p_X(x) p_Y(y)}[e^{T_{\theta}(x,y)}]),
\end{align}
where $\theta \in \Theta$ parameterizes a family of functions $T_{\theta} : \mathcal{X}\times \mathcal{Y} \to \mathbb{R}$ via a deep neural network. MINE is a biased estimator, consequently, the authors suggested to replace the expectation in the denominator of the gradient with an exponential moving average. However, MINE suffers from high-variance estimations and a recent solution named SMILE \cite{Song2020} tries to alleviate such issue by clipping the density ratio between $-e^{-\tau}$ and $e^{\tau}$, with $\tau \geq 0$. SMILE is equivalent to MINE in the limit $\tau \rightarrow +\infty$.

A mutual information unbiased variational lower bound that relies on the $f$-divergence representation was proposed in \cite{Nguyen2010}
\begin{align}
\label{eq:NWJ}
I(X;Y) \geq & I_{NWJ}(X;Y) \nonumber \\
= & \sup_{\theta \in \Theta} \mathbb{E}_{(x,y)\sim p_{XY}(x,y)}[T_{\theta}(x,y)]  \nonumber \\
& -\mathbb{E}_{(x,y)\sim p_X(x) p_Y(y)}[e^{T_{\theta}(x,y)-1}].
\end{align}
where it is easy to verify that $I_{NWJ}\leq I_{MINE}$.

In \cite{LetiziaNIPS}, the discriminative mutual information estimator (DIME) was proposed. It is a neural estimator that mimics the task of the discriminator in GANs \cite{Goodfellow2014}. Indeed, the optimal discriminator in the GAN or $f$-GAN framework can be used either as a direct or indirect mutual information estimator when it is fed with paired and unpaired input data of distribution $p_{XY}$ and $p_X p_Y$, respectively. We refer to as direct estimator, the one where the optimal discriminator $D^*$ is proportional to the density ratio in \eqref{eq:density_ratio}. An example is the d-DIME estimator, a mutual information variational lower bound that reads as follows
\begin{equation}
\label{eq:d_DIME}
I(X;Y) \geq \tilde{I}_{dDIME}(X;Y) = \frac{\mathcal{J}_{\alpha}(D^*)}{\alpha}+1-\log(\alpha),
\end{equation}
where $\mathcal{J}_{\alpha}(D)$, $\alpha>0$, is a value function defined as
\begin{align}
\mathcal{J}_{\alpha}(D) = \; & \alpha \cdot \mathbb{E}_{(x,y) \sim p_{XY}(x,y)}\biggl[\log \biggl(D\bigl(x,y\bigr)\biggr)\biggr] \nonumber \\ 
& +\mathbb{E}_{(x,y) \sim p_{X}(x)p_{Y}(y)}\biggl[-D\bigl(x,y\bigr)\biggr],
\label{eq:discriminator_function}
\end{align}
and
\begin{equation}
\label{eq:optimal_discriminator_2}
D^*(x,y) = \alpha \cdot \frac{p_{XY}(x,y)}{p_{X}(x)\cdot p_Y(y)} = \arg \max_D \mathcal{J}_{\alpha}(D).
\end{equation}
Notice that when the positive function $D(\cdot)$ is parameterized by a neural network, the NWJ estimator can be thought as a special case of d-DIME with $\alpha=1$ and $T_{\theta}(x,y) = \log(D_{\theta}(x,y))+1$.

In the following section we extend DIME, in particular, we introduce $f$-DIME, a discriminative mutual information estimator based on the $f$-divergence measure.   

\section{$f$-DIME}
\label{sec:f-DIME}
In this section, we describe a general methodology to estimate the mutual information by applying the variational representation of $f$-divergence functionals $D_f(P||Q)$. In detail, if the measures $P$ and $Q$ are absolutely continuous w.r.t. $\diff x$ and possess densities $p$ and $q$, then the $f$-divergence reads as follows
\begin{equation}
D_f(P||Q) = \int_{\mathcal{X}}{q(x)f\biggl(\frac{p(x)}{q(x)}\biggr)\diff x},
\end{equation}
where $\mathcal{X}$ is a compact domain and the function $f:\mathbb{R}_+ \to \mathbb{R}$ is convex, lower semicontinuous and satisfies $f(1)=0$. Under these conditions, the authors of \cite{Nguyen2010} (and \cite{Nowozin2016} for the $f$-GAN) exploited the Fenchel convex duality to derive a lower bound on $D_f$
\begin{equation}
\label{eq:f_bound}
D_f(P||Q) \geq \sup_{T\in \mathbb{R}} \biggl\{ \mathbb{E}_{x \sim p(x)} \bigl[T(x)\bigr]-\mathbb{E}_{x\sim q(x)}\bigl[f^*\bigl(T(x)\bigr)\bigr]\biggr\},
\end{equation}
where $T: \mathcal{X} \to \mathbb{R}$ and $f^*$ is the Fenchel conjugate of $f$ defined as
\begin{equation}
f^*(t) := \sup_{u\in \mathbb{R}} \{ ut -f(u)\}.
\end{equation}
Therein, it was shown that the bound in \eqref{eq:f_bound} is tight for optimal values of $T(x)$ and it takes the following form
\begin{equation}
\label{eq:optimal_ratio}
T^*(x) = f^{\prime} \biggl(\frac{p(x)}{q(x)}\biggr),
\end{equation}
where $f^{\prime}$ is the derivative of $f$.

To create a novel neural estimator based on the representation in \eqref{eq:f_bound}, we firstly substitute the mutual information $I(X;Y)$ with its KL divergence representation 
\begin{equation}
I(X;Y) = D_{KL}(p_{XY}||p_X p_Y),
\end{equation}
and by defining 
\begin{align}
\mathcal{J}_{f}(T) = \; & \mathbb{E}_{(x,y) \sim p_{XY}(x,y)}\bigl[T\bigl(x,y\bigr)\bigr] \nonumber \\ 
& -\mathbb{E}_{(x,y) \sim p_{X}(x)p_{Y}(y)}\bigl[f^*\bigl(T\bigl(x,y\bigr)\bigr)\bigr],
\label{eq:value_function_f}
\end{align}
we propose to parametrize $T(x,y)$ with a deep neural network $T_{\theta}$ of parameters $\theta$ and solve with gradient ascent and back-propagation
\begin{equation}
\theta^* = \arg \max_{\theta} \mathcal{J}_f(T_{\theta}).
\end{equation}
The main idea is the following: \textit{it is possible to exploit the relation \eqref{eq:optimal_ratio} to estimate the mutual information}. Indeed, if the derivative of $f$ is invertible, 
\begin{equation}
\bigl(f^{\prime}\bigr)^{-1}\bigl(T_{\theta^*}(x,y)\bigr) = \frac{p_{XY}(x,y)}{p_X(x)p_Y(y)},
\end{equation}
thus, having access to $T_{\theta^*}$ and using the density ratio estimated above, we obtain $f$-DIME as
\begin{equation}
\label{eq:f_dime}
I_{fDIME}(X;Y) = \mathbb{E}_{(x,y) \sim p_{XY}(x,y)}\biggl[ \log \biggl(\bigl(f^{\prime}\bigr)^{-1}\bigl(T_{\theta^{*}}(x,y)\bigr) \biggr) \biggr].
\end{equation}

We argue that any value function $\mathcal{J}_f$ of the form in \eqref{eq:value_function_f}, seen as dual representation of a certain $f$-divergence $D_f$, can be maximized to estimate the mutual information using \eqref{eq:f_dime}. As a notable example, the GAN-based i-DIME estimator in \cite{LetiziaNIPS} can be derived using $f$-DIME with generator 
\begin{equation}
f(u) = u\log u-(u+1)\log(u+1)+\log4
\end{equation}
and $T(x)=\log(D(x))$.

However, only a subset of $\mathcal{J}_f$ can be used inside the capacity-driven autoencoder loss function since a maximization of the mutual information is also required. In particular, only the generator $f$ of the KL divergence is suitable for the mutual information maximization, since in such case
\begin{equation}
\max_{p_X(x)} D_f = \max_{p_X(x)} D_{KL} = \max_{p_X(x)} I(X;Y).
\end{equation}

The generator function of the KL divergence is $f(u) = u\log u$ with conjugate $f^*(t) = e^{t-1}$. It easy to notice that such choice of $f$ coincides with the NWJ estimator or d-DIME when $T(x) = \log(D(x))$ and $\alpha=1$ (see \eqref{eq:d_DIME}). Nevertheless, as shown in \cite{LSGAN}, the value function plays a fundamental role during the training process of a discriminator. Consequently, we combined the approach described for the derivation of $f$-DIME with d-DIME to obtain a family of lower bounds on the mutual information that may result in more robust estimators according to specific case studies.

\begin{lemma}
\label{lemma:Lemma1}
Let $X\sim p_X(x)$ and $Y\sim p_{Y}(y|x)$ be the channel input and output, respectively. Let $D(\cdot)$ be a positive function. If $\mathcal{J}_{\gamma}(D)$, $\gamma>0$, is a value function defined as 
\begin{align}
\mathcal{J}_{\gamma}(D) = \; & \gamma \cdot \mathbb{E}_{(x,y) \sim p_{XY}(x,y)}\biggl[\log \biggl(D\bigl(x,y\bigr)\biggr)\biggr] \nonumber \\ 
& + \mathbb{E}_{(x,y) \sim p_{X}(x)p_{Y}(y)}\biggl[- D^{\gamma}\bigl(x,y\bigr)\biggr],
\end{align}
then
\begin{equation}
I(X;Y) \geq \tilde{I}_{\gamma DIME}(X;Y) = \mathcal{J}_{\gamma}(D^*)+1,
\end{equation}
where
\begin{equation}
\label{eq:optimal_discriminator_gamma}
D^*(x,y) = \biggl(\frac{p_{XY}(x,y)}{p_{X}(x)\cdot p_Y(y)}\biggr)^{1/\gamma} = \arg \max_D \mathcal{J}_{\gamma}(D).
\end{equation}
\end{lemma}
\begin{figure}
\centering
  	\includegraphics[scale = 0.255]{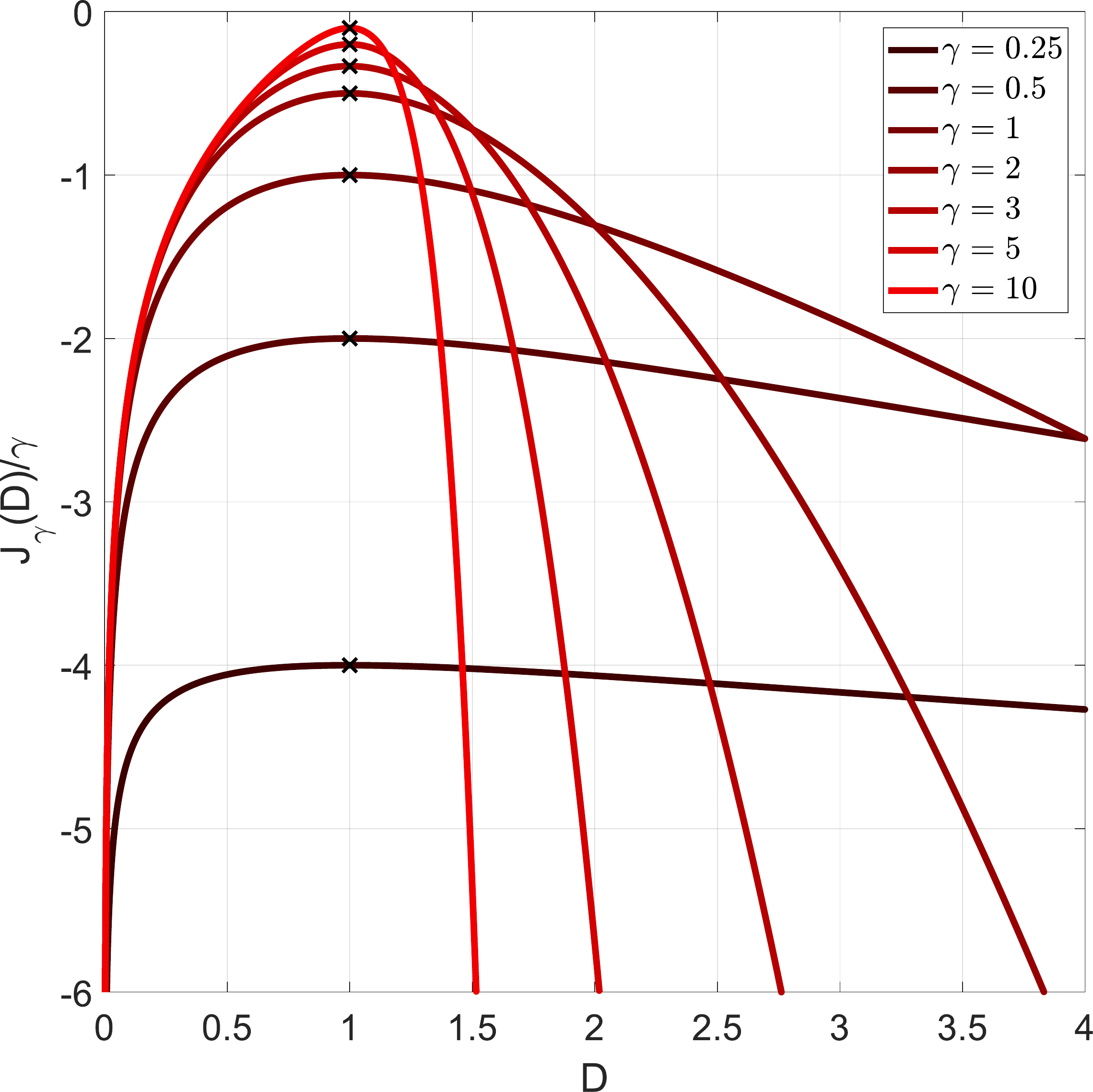}
  	\caption{Behaviour of the functional $J_{\gamma}$ varying the parameter $\gamma$. Concavity of the neighbourhood of the maximal value depends on $\gamma$.}
  	\label{fig:gammaDIME}
\end{figure}

Fig.~\ref{fig:gammaDIME} illustrates how the value function in \eqref{eq:value_function_f} (divided by $\gamma$) behaves according to the choice of the positive parameter $\gamma$. The position of the maximum depends on the density ratio as described in \eqref{eq:optimal_discriminator_gamma} ($R = 1$ in the figure). The role of $\gamma$ is to modify the concavity of the value function to improve the convergence rate of the gradient ascent algorithm. Moreover, when $D$ is parameterized by a neural network, the choice of the learning rate and activation functions may be related with $\gamma$. Indeed, notice that $\gamma$-DIME can also be derived from d-DIME using the substitution $D(x,y) \rightarrow D^{\gamma}(x,y)$.

\section{Results}
\label{sec:results}
In the experimental results, we consider the transmission of $M$ messages at rate $R$ over an AWGN channel, for which we know the channel capacity under an
average power constraint. The choice of AWGN channel shall not be seen as a limitation since the formulation transcends the channel characteristics. Moreover, it provides a reference curve for the validation of the estimators and it allows comparisons across several application domains \cite{Mine2018, LetiziaNIPS}. 

We describe two scenarios: a) the transmission at a rate $R>1$ and short code-length $n$ and b) the transmission at a low rate with larger code-length. For both scenarios, we train a rate-driven autoencoder (see \cite{Letizia2021}) and compare its performance for different estimators in terms of BLER and accuracy of the mutual information estimation. 

For all the experiments and as a proof of concept, we use simple encoder, decoder and neural mutual information estimator architectures. In particular, the encoder is a shallow neural network with $M$ neurons in the input layer, $M$ in the hidden one and $2n$ neurons at the output while the decoder is its complementary.  
The discriminative estimator possesses instead a fixed architecture, a
two layer multilayer perceptron neural network of $200$
neurons in each layer and LeakyReLU with a negative slope coefficient of $0.2$ as activation function. The only
difference in the estimators resides in the final layer where we use a linear dense layer for MINE, while a softplus layer for the DIME based estimators, defined as
\begin{equation}
 D^*(x,y; \theta) = \log(1+\exp(W^*_L \cdot z_{L-1}+b^*_L)),
\end{equation}
where $W$ and $b$ are weights and biases, respectively, $z$ is the output of the previous layer and $L$ is the number of layers of the network.
\begin{figure}
\centering

  	\includegraphics[scale=0.245]{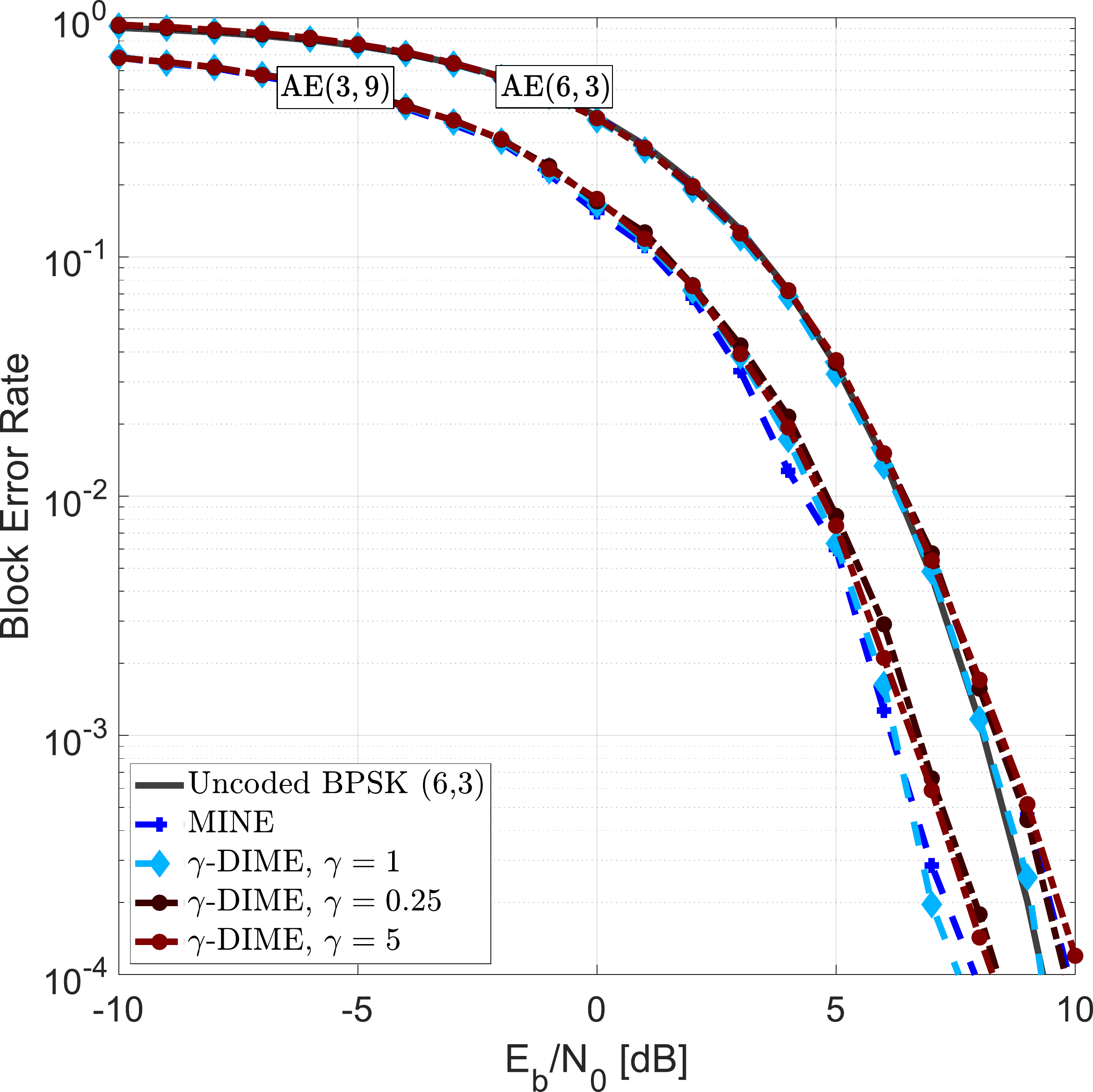}
  	\caption{Comparison of the BLER obtained by the autoencoder AE($3,9$) and AE($6,3$) with an AWGN intermediate layer, using different estimators with $\beta = 0.2$, $\epsilon = 0.2$.}
  	\label{fig:ber_all}
\end{figure}
The training of the autoencoder was conducted at a fixed $E_b/N_0$ ratio of $7$ dB. We minimized the loss function alternating encoder and decoder weights update with the discriminator weights update. The number of iterations was set to $10$k with a learning rate of $0.01$ for both the components. For more implementation details, we refer to the code publicly available. \footnote{\url{https://github.com/tonellolab/capacity-approaching-autoencoders}}

The first scenario describes a digital transmission scheme with input alphabet size $M=64$ over $3$ channel uses. We denote such system as AE($6,3$) since $k = \log_2(M)$. The communication rate is $R=2$. Fig.~\ref{fig:ber_all} shows the autoencoder performance in terms of BLER for different estimators. From the curves we can infer that the $\gamma$-DIME estimator with $\gamma=1$ provides improved performance in terms of BLER compared to both MINE and other values of $\gamma$. The result of Fig.~\ref{fig:MI_6_6} is insightful: the $\gamma$-DIME estimators provide stable estimations also at high SNRs where it is expected that the mutual information saturates to the information rate $R=2$. Thus, the DIME based estimators are more accurate compared to the divergent MINE (see theorem 2 of \cite{Song2020}). It is also interesting to notice that the concavity of $J_{\gamma}$ (shown in Fig.~\ref{fig:gammaDIME}) influences the mutual information estimation as it appears in the results. Indeed, for low values of $\gamma$, fixed learning rate and training iterations, the estimators are trained by moving on the flat curves of Fig.~\ref{fig:gammaDIME}, which is equivalent as using a small learning rate to upgrade the gradients. For completeness and to show the advantage of the proposed estimators under fading, we report in Fig.~\ref{fig:MI_6_6_Ray} the achieved estimated information rate with Rayleigh fading. Also in this case, MINE diverges for high SNR values while the DIME based estimators are much more accurate and tend to saturate to the rate $R=2$. 

\begin{figure}
\centering

  	\includegraphics[scale=0.245]{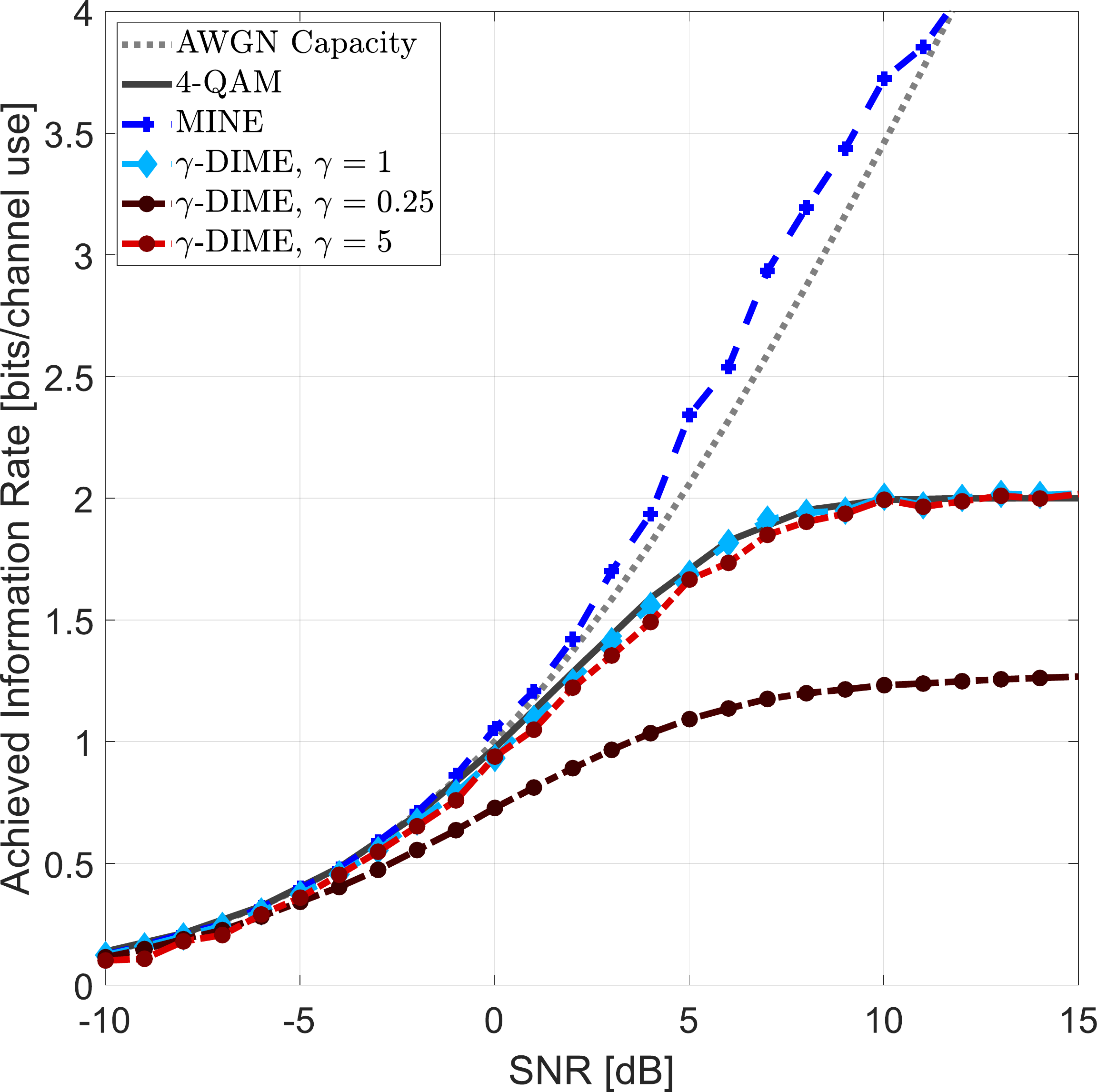}
  	\caption{Comparison of different mutual information estimators for the autoencoder AE($6,3$) with $\beta = 0.2$, $\epsilon = 0.2$ on the AWGN channel.}
  	\label{fig:MI_6_6}
\end{figure}

\begin{figure}
\centering

  	\includegraphics[scale=0.245]{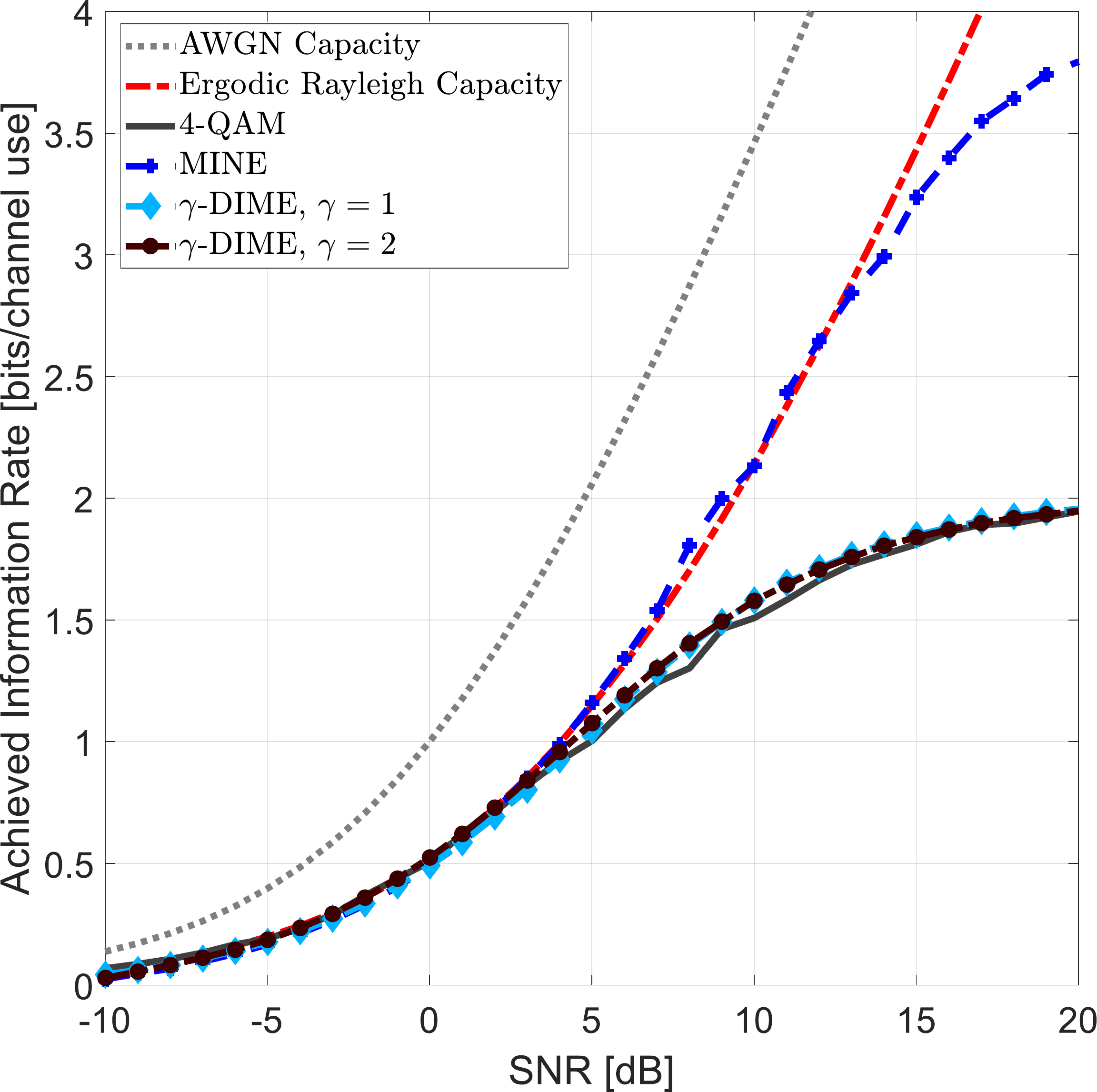}
  	\caption{Comparison of different mutual information estimators for the autoencoder AE($6,3$) with $\beta = 0.2$, $\epsilon = 0.2$ on the Rayleigh fading channel.}
  	\label{fig:MI_6_6_Ray}
\end{figure}

The second scenario describes a digital transmission scheme with input alphabet size $M=8$ over $9$ channel uses. We denote such system as AE($3,9$). The communication rate is $R=1/3$. Fig.~\ref{fig:ber_all} shows that MINE and $\gamma$-DIME with $\gamma=1$ have similar performance in terms BLER under the same training iterations. However, Fig.~\ref{fig:MI_3_18} highlights the fact that even for a low alphabet dimension and for relatively small $n$, the DIME-based estimators are more accurate than MINE since they tend to saturate to $R=1/3$ while MINE overestimates the mutual information. We can conclude that despite MINE being a tighter bound on the mutual information compared to $\gamma$-DIME, the latter and in particular the estimator with $\gamma=1$ performs much better than MINE in terms of mutual information estimation. We leave comments about the architecture to use for future studies, although we remark that it is one of the main aspects in modern machine learning applications and therefore we expect significant improvements. 

\begin{figure}
\centering

  	\includegraphics[scale=0.245]{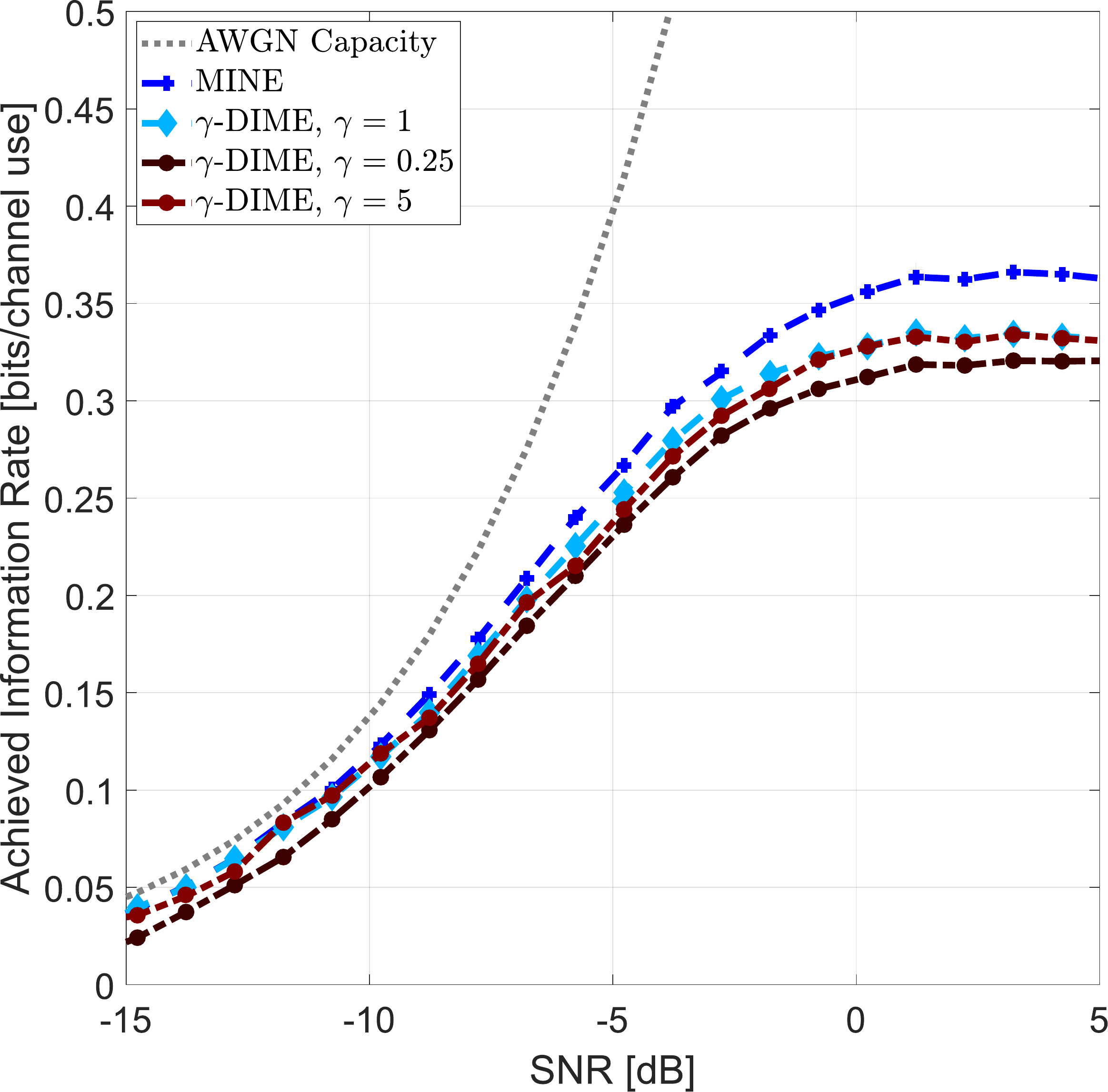}
  	\caption{Comparison of the achieved information rate using different mutual information estimators for the autoencoder AE($3,9$) with $\beta = 0.2$, $\epsilon = 0.2$ on the AWGN channel.}
  	\label{fig:MI_3_18}
\end{figure}

\section{Conclusions}
\label{sec:conclusions}
In this paper, we presented $f$-DIME, a set of novel mutual information estimators based on the $f$-divergence. When $f$ is the function generator of the KL divergence, we obtained as a special case $\gamma$-DIME, a parametric estimator that can be used both for mutual information estimation and design of optimal channel coding. We integrated $\gamma$-DIME as the mutual information regularization term inside the autoencoder loss function. A comparison in terms of estimation accuracy and BLER was done with MINE, an established neural mutual information estimator. Results show the advantage of adopting $\gamma$-DIME over MINE for the autoencoder end-to-end training especially when the aim is to estimate the achieved information rate. 
\appendix
\subsection{Lemmas and Proofs}
\label{sec:omitted_proofs}

 \setcounter{lemma}{0}
 \setcounter{theorem}{0}

\begin{lemma}
Let $X\sim p_X(x)$ and $Y\sim p_{Y}(y|x)$ be the channel input and output, respectively. Let $D(\cdot)$ be a positive function. If $\mathcal{J}_{\gamma}(D)$, $\gamma>0$, is a value function defined as 
\begin{align}
\mathcal{J}_{\gamma}(D) = \; & \gamma \cdot \mathbb{E}_{(x,y) \sim p_{XY}(x,y)}\biggl[\log \biggl(D\bigl(x,y\bigr)\biggr)\biggr] \nonumber \\ 
& + \mathbb{E}_{(x,y) \sim p_{X}(x)p_{Y}(y)}\biggl[- D^{\gamma}\bigl(x,y\bigr)\biggr],
\end{align}
then
\begin{equation}
I(X;Y) \geq \tilde{I}_{\gamma DIME}(X;Y) = \mathcal{J}_{\gamma}(D^*)+1,
\end{equation}
where
\begin{equation}
D^*(x,y) = \biggl(\frac{p_{XY}(x,y)}{p_{X}(x)\cdot p_Y(y)}\biggr)^{1/\gamma} = \arg \max_D \mathcal{J}_{\gamma}(D).
\end{equation}
\end{lemma}

\begin{proof}
Consider a scaled generator $f(u) = \frac{u}{\gamma}\log u$ and for simplicity of notation, denote $p_{XY}$ and $p_Xp_Y$ with $p$ and $q$, respectively. Then
\begin{equation}
D_{KL}(p||q) = \gamma \int_{x}{q(x) \frac{p(x)}{\gamma q(x)}\log\biggl(\frac{p(x)}{q(x)}\biggr) \diff x},
\end{equation}
and the conjugate $f^*(t)$, with $t \in \mathbb{R}$, is given by
\begin{equation}
f^*(t) = \frac{e^{\gamma t -1}}{\gamma}.
\end{equation}
Substituting in \eqref{eq:f_bound} yields to
\begin{equation}
D_{KL}(p||q) \geq \sup_{T\in \mathbb{R}} \biggl\{ \gamma \mathbb{E}_{x \sim p(x)} \bigl[T(x)\bigr]-\mathbb{E}_{x\sim q(x)}\bigl[e^{\gamma T(x)-1} \bigr]\biggr\}.
\end{equation}
Using \eqref{eq:optimal_ratio} it is easy to verify that the optimal value of $T$ is the log-likelihood ratio rather than the density ratio. Indeed,
\begin{equation}
T^*(x) = \frac{1}{\gamma}\log\biggl(\frac{p(x)}{q(x)}\biggr) + \frac{1}{\gamma}.
\end{equation}
Finally, with the change of variable $T(x)=\log(D(x))+1/{\gamma}$, the optimal discriminator has form
\begin{equation}
D^*(x) = \biggl(\frac{p(x)}{q(x)}\biggr)^{1/\gamma}
\end{equation}
and
{
\small
\begin{equation}
D_{KL}(p||q) \geq 1 + \sup_{D\in \mathbb{R}_+} \gamma \mathbb{E}_{x \sim p(x)} \bigl[\log \bigl( D(x) \bigr) \bigr] \nonumber \\
 -\mathbb{E}_{x\sim q(x)}\bigl[D^{\gamma}(x) \bigr].
\end{equation}
}
Denoting the two terms inside the supremum with $J_{\gamma}(D)$ and replacing $p$ and $q$ with the joint and marginal distributions, concludes the proof.
\end{proof}

\bibliographystyle{IEEEtran}
\bibliography{IEEEabrv,biblio}

\end{document}